\newif\ifarxiv\arxivtrue
\newif\ifsmallfigs\smallfigsfalse

\documentclass[11pt]{article}
\usepackage{comment}

\usepackage{caption}

\usepackage{subcaption}
\usepackage{csquotes}
\usepackage{lmodern}

\ifarxiv
\usepackage[utf8]{inputenc}
\usepackage{graphicx}
\usepackage{amssymb,amsmath,amsthm,amsfonts}
\usepackage{amsfonts}
\usepackage[usenames,dvipsnames]{xcolor}
\usepackage{soul}
\usepackage{fullpage}
\usepackage{url}
\usepackage{hyperref}
\usepackage{titlesec}
\titlespacing\section{0pt}{12pt plus 4pt minus 1pt}{0pt plus 2pt minus 1pt}
\titlespacing\subsection{0pt}{12pt plus 4pt minus 1pt}{0pt plus 2pt minus 1pt}
\titlespacing\subsubsection{0pt}{12pt plus 4pt minus 1pt}{0pt plus 2pt minus 1pt}

\else
\usepackage{maa-monthly}
\fi
\usepackage[square, comma, numbers,sort&compress]{natbib}
\usepackage{cleveref}
\usepackage{blindtext}
\usepackage{csquotes}
\usepackage{pgf,tikz,pgfplots}
\pgfplotsset{compat=1.14}
\usepackage{mathrsfs}
\usetikzlibrary{arrows}
\usepackage{tikz-3dplot}

\usepackage[inline]{asymptote}

\newtheorem{claim}{Claim}
\newtheorem{fact}{Fact}
\newtheorem{observation}{Observation}
\newtheorem{lemma}{Lemma}
\newtheorem{theorem}{Theorem}

\newtheorem*{claim-non}{Claim}
\newtheorem*{theorem-non}{Theorem}

\newtheorem{corollary}{Corollary}

\theoremstyle{definition}
\newtheorem{definition}{Definition}
\theoremstyle{remark}

\newcommand{\mute}[1]{}

\newcommand{\mc}[1]{\mathcal{#1}}
\newcommand{\mbb}[1]{\mathbb{#1}}
\newcommand{\R}{\ensuremath{\mathbb{R}}\xspace}

\newcommand{\ssm}{\ensuremath{{\smallsetminus}}\xspace}

\newcommand{\vphi}{\varphi}

\newif \ifdraft \drafttrue

\newcommand{\cit}[1]{\textcolor{Orange}{[citation needed]}}

\title{The Gerrymandering Jumble: Map Projections Permute Districts' Compactness Scores}
\ifarxiv
\author{Assaf Bar-Natan\thanks{University of Toronto, Department of Mathematics. \href{mailto:assaf.barnatan@gmail.com}{\texttt{assaf.barnatan@gmail.com}}} 
	\and Lorenzo Najt\thanks{University of Wisconsin, Madison, Department of Mathematics. \href{mailto:lnajt4@gmail.com}{\texttt{lnajt@math.wisc.edu}}}  
	\and Zachary Schutzman\thanks{University of Pennsylvania, Department of Computer and Information Science. \href{mailto:ianzach@seas.upenn.edu}{\texttt{ianzach@seas.upenn.edu}}}  }
\else
\author{Anonymous Author(s)}
\fi
\begin{document}
\maketitle
\begin{abstract}
\noindent In political redistricting, the \textit{compactness} of a district is
used as a quantitative proxy for its fairness.  Several
well-established, yet competing, notions of geographic compactness are
commonly used to evaluate the shapes of regions, including the
\textit{Polsby-Popper score}, the \textit{convex hull score}, and the \textit{Reock score}, and
these scores are used to compare two or more districts or plans.  In
this paper, we prove mathematically that  any \textit{map
projection} from the sphere to the plane reverses the ordering of the scores of some pair of regions for all three of these scores.  Empirically, we demonstrate that the effect of using the Cartesian latitude-longitude projection on the order of Reock scores is quite dramatic.

\end{abstract}
\ifarxiv
\else
%{\small\paragraph{AMS Subject Classifications:} 	51N15	\ \ \  97A40}
\fi

\section{Introduction}
Striving for the \textit{geographic compactness} of electoral
districts is a traditional principle of redistricting \cite{altman_1998}, and, to that
end, many jurisdictions have included the criterion of compactness in
their legal code for drawing districts.  Some of these include Iowa's\footnote{Iowa Code \S42.4(4)} measuring the perimeter of districts, Maine's\footnote{Maine Statute \S1206-A} minimizing travel time within a district, and Idaho's\footnote{Idaho Statute 72-1506(4)} avoiding 
drawing districts which are \enquote*{oddly shaped}.  Such measures can vary widely in their 
precision, both mathematical and otherwise.  Computing the perimeter of districts is a very clear definition, minimizing travel time is less so, and what makes a district oddly shaped or not seems rather challenging to consider from a rigorous standpoint. 

While a strict definition of when a district is or is not \enquote*{compact} is quite elusive, the purpose of such a criterion is much easier to articulate.  Simply put, a district which is bizarrely shaped, such as one with small tendrils grabbing many distant chunks of territory,
probably wasn't drawn like that by accident. Such a shape need not be drawn for nefarious
purposes, but its unusual nature should trigger closer scrutiny.  Measures to compute the
geographic compactness of districts are intended to formalize this quality of \enquote*{bizarreness}
mathematically.  We briefly note here that the term \textit{compactness} is somewhat 
overloaded, and that we exclusively use the term to refer to the shape of geographic 
regions and not to the topological definition of the word.

People have formally studied geographic compactness for nearly two hundred years, and, over that period, scientists and legal scholars have developed many formulas to assign a numerical measure of \enquote*{compactness} to a region such as an electoral district \cite{young_compactness}.
 Three of the most commonly discussed formulations are the \textit{Polsby-Popper score}, which
measures the normalized ratio of a district's area to the square of its
perimeter, the \textit{convex hull score}, which measures the ratio
of the area of a district to the smallest convex region containing it,
and the \textit{Reock score}, which measures the ratio of the area of
a district to the area of the smallest circular disc containing it.  Each of these
measures is appealing at an intuitive level, since they assign to
a district a single scalar value between zero and one, which presents a simple 
method to compare the relative compactness of two or more districts. 
Additionally, the
mathematics underpinning each is widely understandable by the relevant
parties, including lawmakers, judges, advocacy groups, and the general
public.  

However, none of these measures truly discerns which districts are \enquote*{compact} and which are not. 
For each score, we can construct a 
mathematical counterexample for which
a human's intuition and the score's evaluation of a shape's
compactness differ.  A region which is roughly circular but has a jagged boundary 
may appear compact to a human's eye, but such a shape has a very poor Polsby-Popper score.  Similarly, a very long, thin rectangle appears non-compact to a person, but has a perfect convex hull score.  Additionally, these scores often do not agree.
The long, thin rectangle has a  very
poor Polsby-Popper score, and the ragged circle has an excellent convex hull score.  These issues are well-studied by political
scientists and mathematicians alike
\cite{polsby1991third,frolov1975shape,maceachren1985compact,barnes2018gerrymandering}.

In this paper, we propose a further critique of these measures, namely
\textit{sensitivity under the choice of map projection}.  Each of the
compactness scores named above is defined as a tool to evaluate
geometric shapes in the plane, but in reality we are interested in
analyzing shapes which sit on the surface of the planet Earth, which
is (roughly) spherical.  
When we analyze the geometric properties of a geographic region, we work 
with a \textit{projection} of the Earth onto a flat plane, such as a piece of 
paper or the screen of a computer.
Therefore, when a shape is assigned a compactness score,
it is implicitly done with respect to some choice of map projection.
We prove that this may have
serious consequences for the comparison of districts by these scores.  In
particular, we consider the Polsby-Popper, convex hull,
and Reock scores on the sphere, and demonstrate that for any choice of
map projection, there are two regions, $A$ and $B$, such that $A$ is
more compact than $B$ on the sphere but $B$ is more compact than $A$
when projected to the plane.  We prove our results in a theoretical context 
before evaluating the extent of this phenomenon empirically.  We find 
that with real-world examples of Congressional districts, the effect 
of the commonly-used \textit{Cartesian latitude-longitude} projection on the convex 
hull and Polsby-Popper scores is relatively minor, but the impact on 
Reock scores is quite dramatic, which may have serious implications 
for the use of this measure as a tool to evaluate geographic compactness.

\subsection{Organization}

For each of the compactness scores we analyze, our proof that no map
projection can preserve their order follows a similar recipe. We
first use the fact that any map projection which preserves an ordering
must preserve the \textit{maximizers} in that ordering.  In other words,
if there is some shape which a score says is \enquote{the most compact} on the sphere 
but the projection sends this to a shape in the plane which is \enquote{not the most compact}, then whatever 
shape \textit{does} get sent to the most compact shape in the plane leapfrogs 
the first shape in the induced ordering.  For all three of the scores we study, such a maximizer exists.

Using this observation, we can restrict our attention to those map
projections which preserve the maximizers in the induced ordering,
then argue that any projection in this restricted set must permute the
order of scores of some pair of regions.

\paragraph{Preliminaries}
We first introduce some definitions and results which we will use to prove our 
three main theorems.  Since spherical geometry differs from the more familiar planar geometry, 
we carefully describe a few properties of spherical lines and triangles to build some intuition 
in this domain.

\paragraph{Convex Hull}
For the convex hull score, we first 
show that any projection which preserves the maximizers of the convex hull score 
ordering must maintain certain geometric properties of shapes and line segments 
between the sphere and the plane.  Using this, we demonstrate that no 
map projection from the sphere to the plane can preserve these properties, and therefore 
no such convex hull score order preserving projection exists.

\paragraph{Reock}
For the Reock score, we follow a similar tack, first showing that any 
order-preserving map projection must also preserve some geometric properties 
and then demonstrating that such a map projection cannot exist.

\paragraph{Polsby-Popper}
To demonstrate that there is no projection which maintains the score ordering induced by the Polsby-Popper score,  we leverage the 
difference between the \textit{isoperimetric inequalities} on the sphere and in the plane, in that the inequality for the plane is scale invariant in that setting but not on the sphere, in order to find a pair of regions in the sphere, one more compact than the other, such that the less compact one is sent to a circle under the map projection.

\paragraph{Empirical Results}
We finally examine the impact of the Cartesian latitude-longitude map projection on the convex hull, Reock, and Polsby-Popper 
scores and the ordering of regions under these scores.  While the impacts of the projection on the convex hull and 
 Polsby-Popper scores and their orderings are not severe, the Reock score and the Reock score ordering both change dramatically 
 under the map projection.

\section{Preliminaries}\label{sec:prelims}

We begin by introducing some necessary observations, definitions, and terminology
which will be of use later.  
\subsection{Spherical Geometry}

In this section, we present some basic results about spherical geometry with the goal of proving Girard's Theorem, which states that the area of a triangle on the unit sphere is the sum of its interior angles minus $\pi$.  Readers familiar with this result should feel free to skip ahead.

We use $\R^2$ to denote the 
Euclidean plane with the usual way of measuring distances, 
$$d(x,y) = \sqrt{(x-y)^2};$$
similarly, $\R^3$ denotes Euclidean 3-space.  We use $\mbb{S}^2$ to denote the \textit{unit 2-sphere}, which can be 
thought of as the set of points in $\R^3$ at Euclidean distance one from the origin.  
 
In this paper, we only consider the sphere and the plane, and leave the consideration of other surfaces, measures, and metrics to future work.

\begin{definition}
{On the sphere, a \textbf{great circle} is the intersection of the sphere with a plane passing through the origin.} These are the circles on the sphere with radius equal to that of the sphere.  See \Cref{fig:sphereline} for an illustration.
\end{definition}

\begin{definition}

Lines in the plane and great circles on the sphere are called \textbf{geodesics}.  A \textbf{geodesic segment} is a line segment in the plane and an arc of a great circle on the sphere.
	
\end{definition}

{The idea of {geodesics} generalizes the notion of \enquote*{straight lines} in the plane to other settings.} One critical difference is that in the plane, there is a unique line passing through any two distinct points and a unique line segment joining them.  On the sphere, there will typically be a unique great circle and \textit{two} geodesic segments through a pair of points, with the exception of one case.

\begin{definition}
	A \textbf{triangle} in the plane or the sphere is defined by three distinct points and the shortest geodesics connecting each pair of points.
\end{definition}

{
\ifsmallfigs
\else
\begin{figure}[htb]
	\centering
	\includegraphics[width=.5\textwidth]{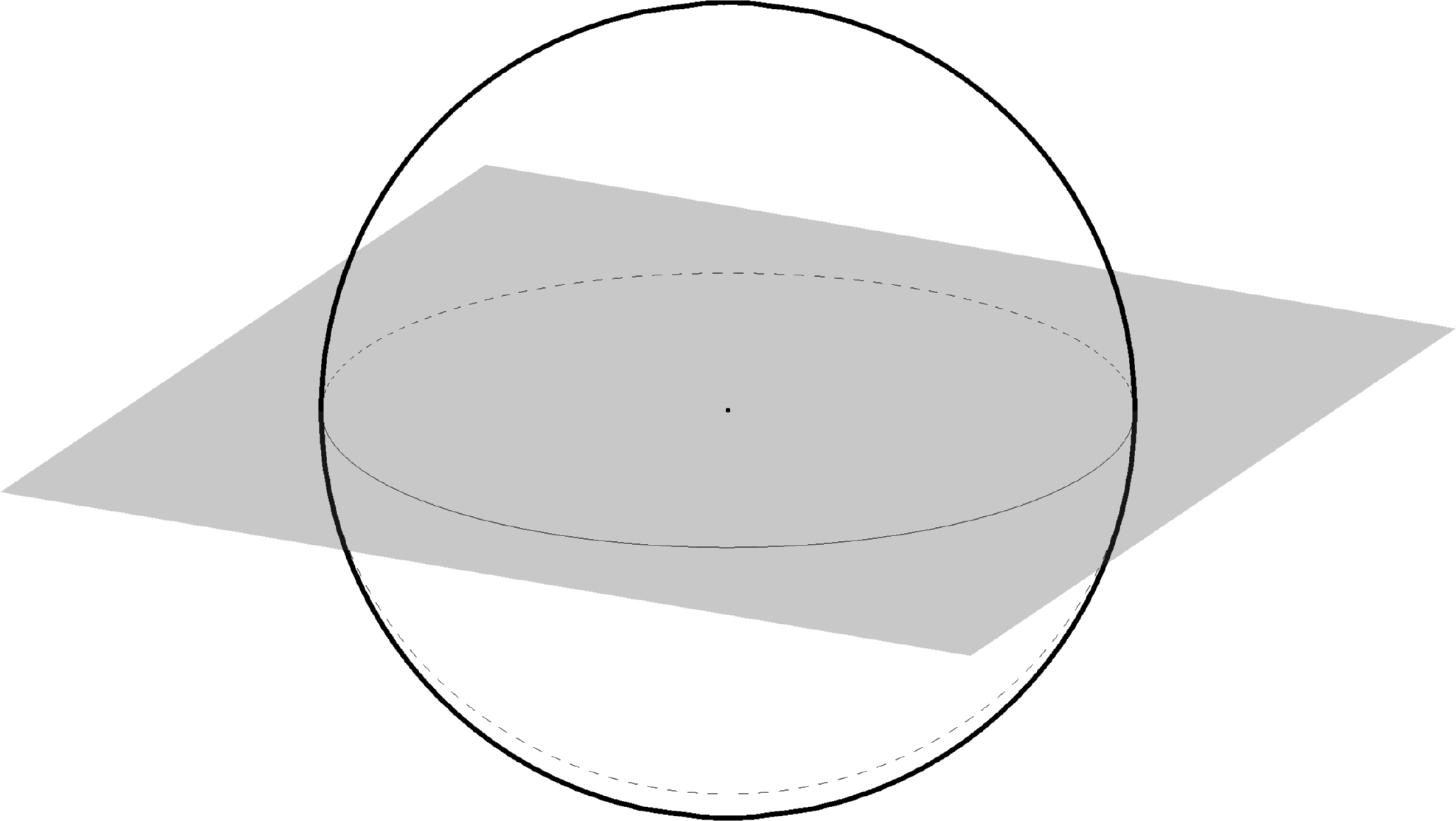}
	\caption{A great circle on the sphere with its identifying plane.}
	\label{fig:sphereline}
\end{figure}
\fi
}

\begin{observation}
	Given any two points $p$ and $q$ on the sphere which are not antipodal, meaning that our points aren't of the form $p=(x,y,z)$ and $q=(-x,-y,-z)$, there is a unique great circle through $p$ and $q$ and therefore two geodesic segments joining them. 
\end{observation}
\mute{
\begin{proof}
	To see this, consider the characterization of great circles as the intersections of the surface of the sphere with planes through the center of the sphere.  Since $p$ and $q$ are not antipodes, they are not both collinear with $(0,0,0)$ and so these three points uniquely determine a plane.  This plane intersects the sphere along a great circle which contains $p$ and $q$.
\end{proof}
}
If $p$ and $q$ \textit{are} antipodal, then any great circle containing one must contain the other as well, so there are infinitely many such great circles. For any two non-antipodal points on the sphere, one of the geodesic segments will be shorter than the other.  {This shorter geodesic segment is the shortest path between the points and its length is the metric distance between $p$ and $q$}.

\ifsmallfigs
\begin{figure}
	\centering
	\begin{minipage}{.5\textwidth}
		\centering
	\includegraphics[height=.58\linewidth]{figs/sph-1pl.pdf}
		\captionof{figure}{A great circle on the sphere with its plane.}
		\label{fig:test1}
	\end{minipage}%
	\begin{minipage}{.5\textwidth}
		\centering
	\includegraphics[height=.58\linewidth]{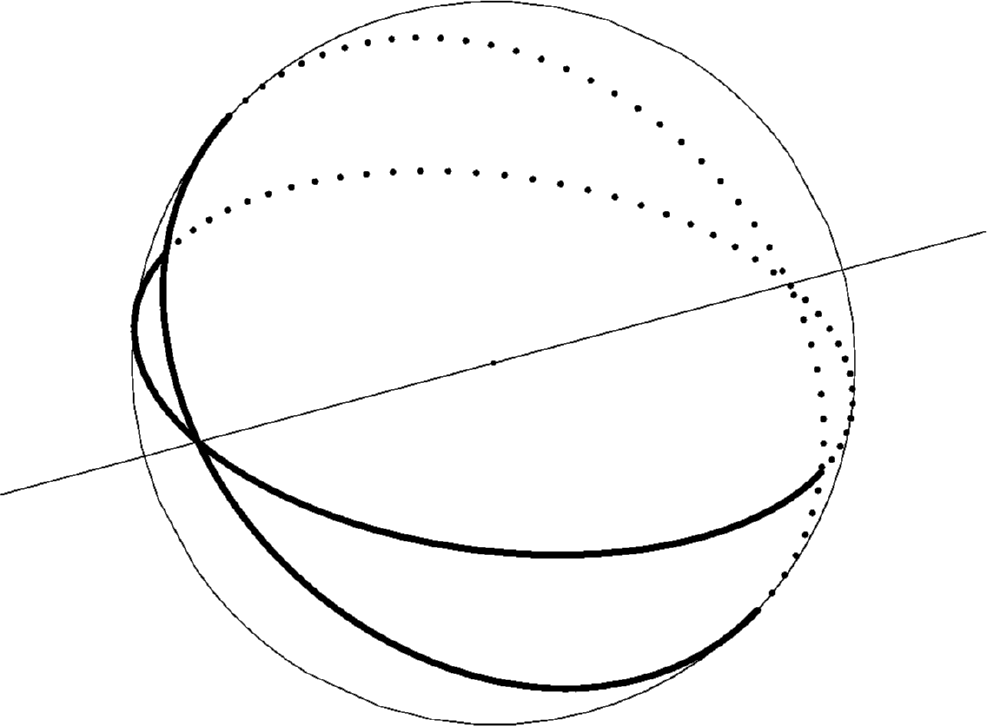}
		\captionof{figure}{Two great circles meet at antipodal points.}
	\label{fig:sphereline}
	\end{minipage}
\end{figure}
\fi

\ifsmallfigs
\else
\begin{figure}[htb]
	\centering
	\includegraphics[width=.35\textwidth]{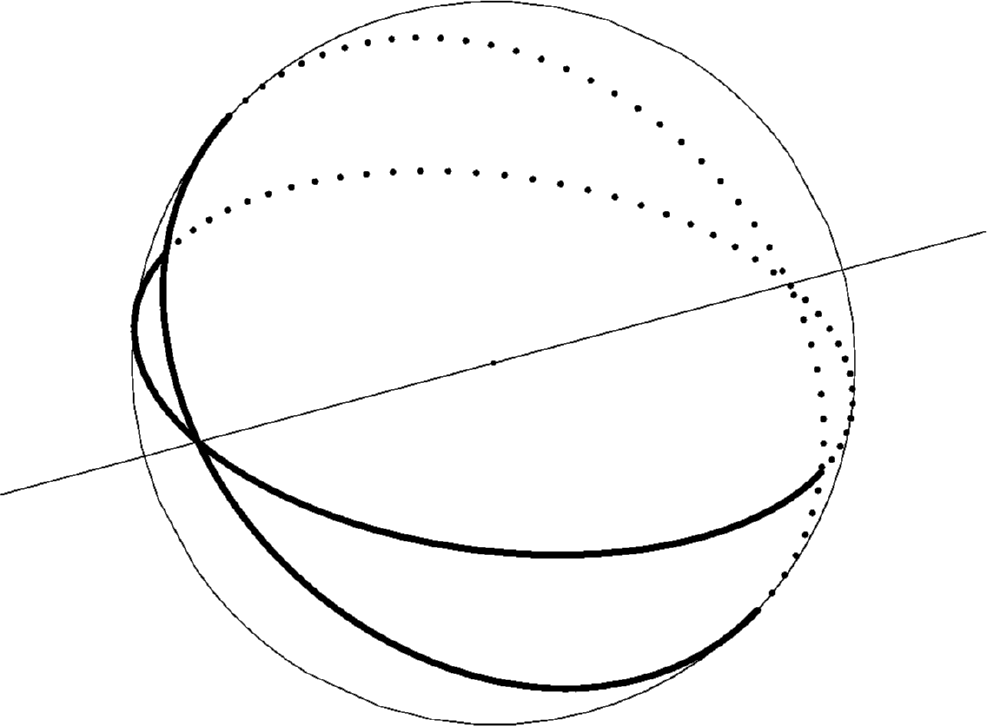}
	\caption{Two great circles meet at antipodal points.}
	\label{fig:2gc}
\end{figure}
\fi

We now have enough terminology to show a very important fact about spherical geometry.  This  observation is one of the salient features which distinguishes it from the more familiar planar geometry.

\begin{claim}
	Any pair of distinct great circles on the sphere intersect exactly twice, and the points of intersection are antipodes.
\end{claim}
\mute{
\begin{proof}
	Any two distinct planes determine two distinct great circles, and these planes intersect along a line in $\R^3$ which passes through $(0,0,0)$ and therefore meets the sphere itself at exactly two points, which must be antipodes.
\end{proof}
}

Why is this weird? In the plane, it is always the case that any pair of distinct lines intersects exactly once or never, in which case we call them \textit{parallel}. Since distinct great circles on the sphere intersect exactly twice, there is no such thing as \enquote*{parallel lines} on the sphere, and we have to be careful about discussing `the' intersection of two great circles since they do not meet at a unique point.  Furthermore, it is not the case that there is a unique segment of a great circle connecting any two points; there are two, but unless our two points are antipodes, one of the two segments will be shorter.

\mute{We will use this distinction to derive a result which we will use critically in the main theorems of this paper, namely\\

\noindent\textbf{\Cref{lem:sphtri}.} [Girard's Theorem]
\emph{The sum of the interior angles of a spherical triangle is strictly greater than $\pi$.  More specifically, the sum of the interior angles is equal to $\pi$ plus the area of the triangle.}\\}

Another difference between spherical and planar geometry appears when computing the angles of triangles. In the planar setting, the sum of the interior angles of a triangle is always $\pi$, regardless of its area. However, in the spherical case we can construct a triangle with three right angles. The north pole and two points on the equator, one a quarter of the way around the sphere from the other, form such a triangle.  Its area is one eighth of the whole sphere, or $\tfrac{\pi}{2}$, which is, not coincidentally, equal to $\tfrac{\pi}{2}+\tfrac{\pi}{2}+\tfrac{\pi}{2} - \pi$. Girard's theorem, which we will prove below, connects the total angle to the area of a spherical triangle.

In order to  show Girard's Theorem, we need some way to translate between \textit{angles} and \textit{area}.  To do that, we'll use a shape which doesn't even exist in the plane: the \textit{diangle} or \textit{lune}.    We know that two great circles intersect at two antipodal points, and we can also see that they cut the surface of the sphere into four regions.  Consider one of these regions.  Its boundary is a pair of great circle segments which connect antipodal points and meet at some angle $\theta\leq \pi$ at both of these points.   

Using that the surface area of a unit sphere is $4\pi$, computing the area of a lune with angle $\theta$ is straightforward.

\begin{figure}[htb]
	\centering
	\includegraphics[width=.35\textwidth]{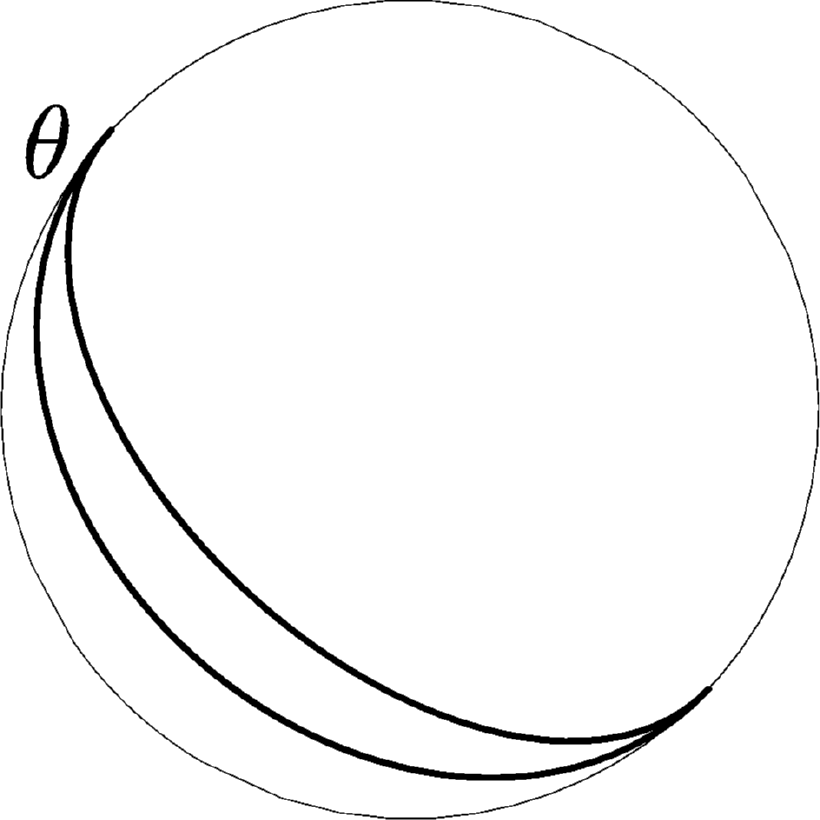}
	\caption{A lune corresponding to an angle $\theta$. }
	\label{fig:lune}
\end{figure}
		\begin{figure}[htb]
	\centering
	\includegraphics[width=.35\textwidth]{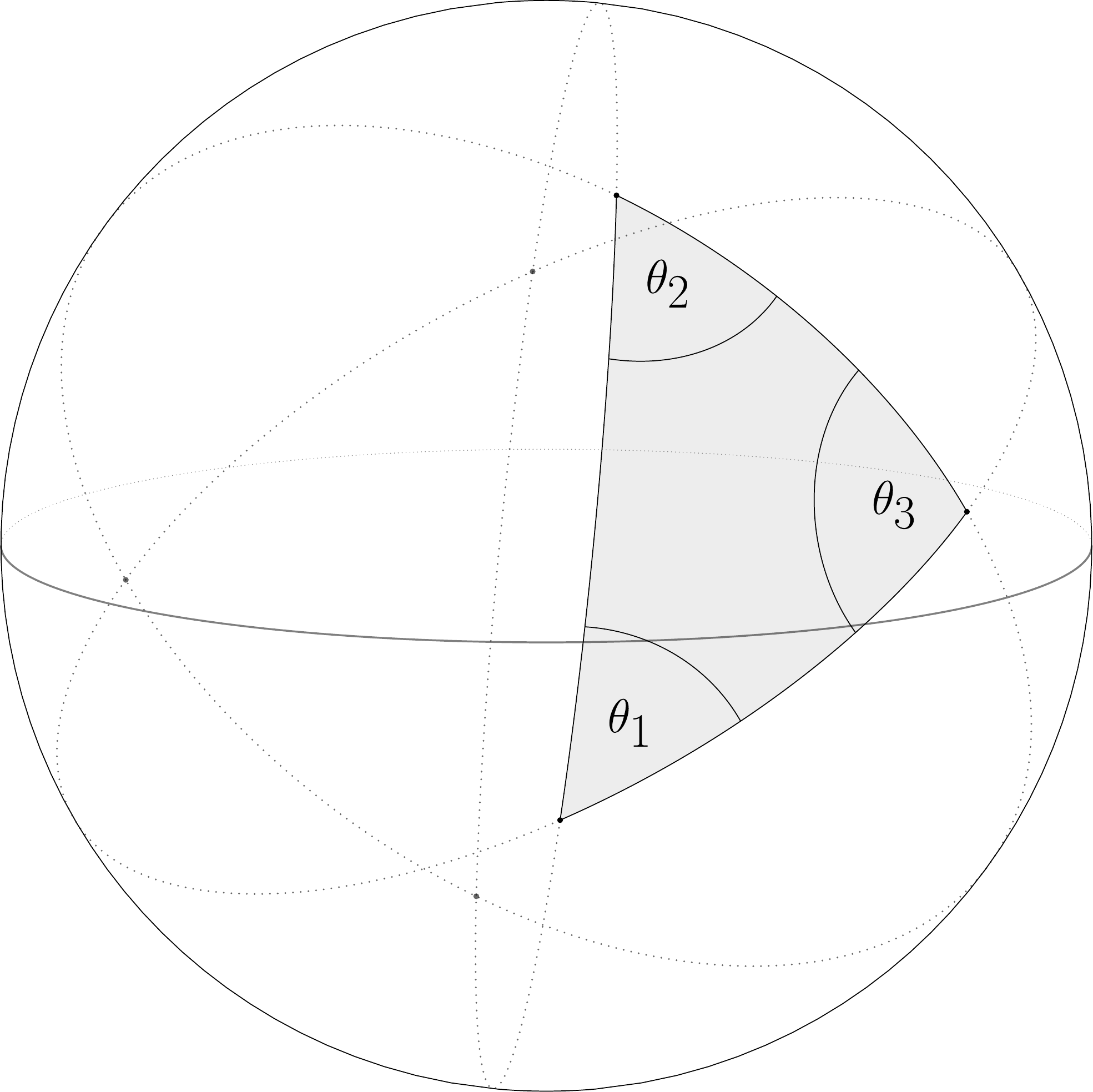}
	\caption{A spherical triangle and the antipodal triangle define six lunes.}
	\label{fig:trilune}
\end{figure}

\begin{claim}
	Consider a lune whose boundary segments meet at angle $\theta$.  Then the area of this lune is $2\theta$.
\end{claim}

Now that we have a tool that lets us relate angles and areas, we can prove Girard's Theorem.

\begin{lemma}(Girard's Theorem)\label{lem:sphtri}
	
	The sum of the interior angles of a spherical triangle is strictly greater than $\pi$.  More specifically, the sum of the interior angles is equal to $\pi$ plus the area of the triangle.
\end{lemma}

\begin{proof}
	Consider a triangle $T$ on the sphere with angles $\theta_1$, $\theta_2$, and $\theta_3$.  Let $\mathrm{area}(T)$ denote the area of this triangle. If we extend the sides of the triangle to their entire great circles, each pair intersects at the vertices of $T$ as well as the three points antipodal to the vertices of $T$, and at the same angles at antipodal points.  This second triangle is congruent to $T$, so its area is also $\mathrm{area}(T)$.  Each pair of great circles cuts the sphere into four lunes, one which contains $T$, one which contains the antipodal triangle, and two which do not contain either triangle.  We are interested in the three pairs of lunes which do contain the triangles.  We will label these lunes by their angles, so we have a lune $L(\theta_1)$ and its antipodal lune $L'(\theta_1)$, and we can similarly define $L(\theta_2)$, $L'(\theta_2)$, $L(\theta_3)$, and $L'(\theta_3)$.

	We have six lunes.  In total, they cover the sphere, but share some overlap.  If we remove $T$ from two of the three which contain it and the antipodal triangle from two of the three which contain it, then we have six non-overlapping regions which cover the sphere, so the area of the sphere must be equal to the sum of the areas of these six regions.  \mute{We can write
	
	\begin{align*}
	4\pi &= \mathrm{area}(L(\theta_1)) + \mathrm{area}(L'(\theta_1)) \\
	&+  (\mathrm{area}(L(\theta_2)) - \mathrm{area}(T)) + (\mathrm{area}(L'(\theta_2)) - \mathrm{area}(T)) \\
	&+ (\mathrm{area}(L(\theta_3)) - \mathrm{area}(T))	 + (\mathrm{area}(L'(\theta_3)) - \mathrm{area}(T)).
	\end{align*}
}

	By the earlier claim, we know that the areas of the lunes are twice their angles, so we can write this as
	
	\ifarxiv
	\begin{align*}
	4\pi &= 2\theta_1 + 2\theta_1 
	+  (2\theta_2 - \mathrm{area}(T)) + (2\theta_2-\mathrm{area}(T))
	+ (2\theta_3 - \mathrm{area}(T))	 + (2\theta_3 - \mathrm{area}(T))
	\end{align*}
	\else
\begin{align*}
4\pi &= 2\theta_1 + 2\theta_1 
+  (2\theta_2 - \mathrm{area}(T)) + (2\theta_2-\mathrm{area}(T))\\
&\phantom{2\pi} + (2\theta_3 - \mathrm{area}(T))	 + (2\theta_3 - \mathrm{area}(T))
\end{align*}	
	\fi
	and rearrange to get

	\begin{align*}
	\theta_1+\theta_2+\theta_3 = \pi + \mathrm{area}(T),
	\end{align*}
	
	which is exactly the statement we wanted to show.
\end{proof}

We will need one more fact about spherical triangles before we conclude this section.  It follows immediately from the Spherical Law of Cosines.

\begin{fact}
	An equilateral triangle is equiangular, and vice versa, where \textit{equilateral} means that the three sides have equal length and \textit{equiangular} means that the three angles all have the same measure.
\end{fact}

\mute{
\begin{proof}
	To see this, first suppose that we have a triangle with vertices $a$, $b$, and $c$ and suppose that the length of the side $ab$ is equal to that of $bc$, and let $m$ be the midpoint of the segment $bc$.  Consider the segment $am$. This splits the triangle $abc$ into two triangles $amb$ and $amc$.  These three triangles have the same side lengths and are therefore congruent.  By this congruence, the angle opposite vertex $a$ is equal to the angle opposite vertex $c$.  To show that either of these is also equal to the angle opposite vertex $b$, take $m$ to be the midpoint of segment $ac$.
	
	To see the converse, suppose that the angles at vertices $b$ and $c$ are equal.  Then the triangle $abc$ is congruent to the triangle $acb$ because they both share side $bc$ and the angle at vertex $a$, and the angle at vertex $b$ is equal to that at $c$.\footnote{This is sometimes called the \enquote{side-side-angle} congruence theorem.}  Since the triangles are congruent, the sides $ac$ and $ab$ have equal length.  To show that either of these is also equal to the length of side $bc$, consider angles $a$ and $c$ instead.
\end{proof}

}

An astute reader may notice that this result is also true of planar triangles, and the planar version follows from Propositions I.6 and I.8 in Euclid's \textit{Elements} \cite{elements,se_triangle}.  Since Euclid's proof doesn't rely on the existence of parallel lines, this fact can alternatively be shown using his argument.

\subsection{Some Definitions}

Now that we have the necessary tools of spherical geometry, we will wrap up this section with a battery of definitions. 
We carefully lay these out so
as to align with an intuitive understanding of the concepts and to
appease the astute reader who may be concerned with edge cases,
geometric weirdness, and nonmeasurability. 
Throughout, we implicitly consider all figures on the sphere to be strictly contained in a hemisphere.

\begin{definition}
	A \textbf{region} is a non-empty, open subset $\Omega$ of $\mathbb{S}^2$ or $\R^2$ such that $\Omega$ is bounded and its boundary is piecewise smooth.
\end{definition}

We choose this definition to ensure that the \textit{area} and \textit{perimeter} of the region are well-defined concepts.  This eliminates pathological examples of open sets whose boundaries have non-zero area or edge cases like considering the whole plane a \enquote*{region}.

\begin{definition}
A \textbf{compactness score function} $\mathcal{C}$ is a function from
the set of all regions to the non-negative real numbers or infinity.  We can compare 
the scores of any two regions, and we adopt the convention that 
\textit{more compact} regions have \textit{higher} scores.  That is,
region $A$ is at least as compact as region $B$ if and only if 
$\mathcal{C}(A)\geq    \mathcal{C}(B)$.
\end{definition}

The final major definition we need is that of a \textit{map
projection}.  In reality, the regions we are interested in comparing
sit on the surface of the Earth (i.e. a sphere), but these regions are
often examined after being projected onto a flat sheet of paper or
computer screen, and so have been subject to such a projection.

\begin{definition}
  A \textbf{map projection} $\varphi$ is a 
  diffeomorphism from a region on the sphere to a region of the 
  plane. 
\end{definition}

We choose this definition, and particularly the term \textit{diffeomorphism}, to ensure that $\vphi$ is smooth, its inverse $\vphi^{-1}$ exists and is smooth, and both $\vphi$ and $\vphi^{-1}$ send regions in their domain to regions in their codomain.  Throughout, we use $\vphi$ to denote such a function from a region of the sphere 
to a region of the plane and $\vphi^{-1}$, to denote the inverse which is a function from a region of the plane back to a region of the sphere.

Since the image of a region under a map projection $\varphi$ is also
a region, we can examine the compactness score of that region both 
before and after applying $\varphi$, and this is the heart of the
problem we address in this paper.  We demonstrate, for several
examples of compactness scores $\mathcal{C}$, that the order
induced by $\mathcal{C}$ is different than the order induced by
$\mathcal{C}\circ\varphi$ for \textit{any} choice of map projection
$\varphi$.

\begin{definition}
  We say that a map projection $\vphi$ \textbf{preserves the  
  compactness score ordering} of a score $\mc{C}$ if for any regions 
  $\Omega,\Omega'$ in the domain of $\vphi$, $\mc{C}(\Omega)\ge \mc{C}(\Omega')$ 
  if and only if $\mc{C}(\vphi(\Omega)) \ge \mc{C}(\vphi(\Omega'))$ in the plane.
\end{definition}

   This is a weaker condition than simply preserving the raw compactness scores. 
   If there is some map projection which results in adding $.1$ to the score of each region, the raw scores are certainly not preserved, but the ordering of regions by their scores is. Additionally, $\vphi$ preserves a compactness score ordering 
  if and only if $\vphi^{-1}$ does.

\begin{definition}
  A 
  \textbf{cap} on the sphere  $\mbb{S}^2$ is a region on the sphere
 which can be described as all of the points on the sphere to one side of some plane 
 in $\R^3$.  A cap has a \textit{height}, which is the largest distance between this cutting plane and the cap, and a \textit{radius}, which is the radius of the circle formed by the intersection of the plane and the sphere.  See \Cref{fig:caphr} for an illustration.
\end{definition}

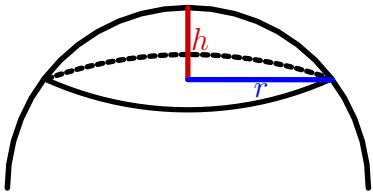
\begin{figure}[h]
  \centering
  \definecolor{qqqqff}{rgb}{0,0,1}

\definecolor{ccqqqq}{rgb}{0.8,0,0}

\definecolor{ududff}{rgb}{0.30196078431372547,0.30196078431372547,1}

\begin{tikzpicture}[scale=.6,line cap=round,line join=round,>=triangle 45,x=1cm,y=1cm]

\clip(-4.493355050909743,-1.2748355780287404) rectangle (4.462150807191927,4.785872976331056);

\draw [shift={(0,0)},line width=2.1pt]  plot[domain=0:3.141592653589793,variable=\t]({1*4*cos(\t r)+0*4*sin(\t r)},{0*4*cos(\t r)+1*4*sin(\t r)});

\draw [shift={(0.0016366799970421299,9.441278259356354)},line width=2pt]  plot[domain=4.286394541961201:5.138764290071874,variable=\t]({1*7.708897306730456*cos(\t r)+0*7.708897306730456*sin(\t r)},{0*7.708897306730456*cos(\t r)+1*7.708897306730456*sin(\t r)});

\draw [shift={(0.019848790872865698,-6.481795867514522)},line width=2pt,dotted]  plot[domain=1.22830273729039:1.9174384984929125,variable=\t]({1*9.441973016351305*cos(\t r)+0*9.441973016351305*sin(\t r)},{0*9.441973016351305*cos(\t r)+1*9.441973016351305*sin(\t r)});

\draw [line width=2pt,color=ccqqqq] (0,2.4033620491027756)-- (0,4);

\draw [line width=2pt,color=qqqqff] (0,2.4033620491027756)-- (3.2001911763834157,2.3997450770024993);

\begin{scriptsize}

\draw [fill=ududff] (0.0016366799970421299,9.441278259356354) circle (2.2pt);

\draw [fill=ududff] (0.019848790872865698,-6.481795867514522) circle (2.2pt);

\draw[color=ccqqqq] (0.26805571312169243,3.3034918797019284) node {\large$h$};

\draw[color=qqqqff] (1.6285727817769453,2.167711482419824) node {\large $r$};

\end{scriptsize}

\end{tikzpicture}
  \caption{ The height $h$ and radius $r$ of a spherical cap. }
  \label{fig:caphr}
\end{figure}

\section{Convex Hull}\label{sec:ch}
We first consider the \textit{convex hull
score}.  We briefly recall the definition of a convex set and then
define this score function.

\begin{definition}
	A set in $\R^2$ or $\mbb{S}^2$  is \textbf{convex} if every shortest geodesic segment between any two points in the set is entirely 
	contained within that set.
\end{definition}

\begin{definition}
  Let $\mathrm{conv}(\Omega)$ denote the \textit{convex hull} of
  a region $\Omega$ in either the sphere or the plane, which is the
  smallest convex region containing $\Omega$.  Then we define the
  \textit{convex hull score} of $\Omega$ as 
  \begin{align*}
    \mathrm{CH}(\Omega)=
    \frac{\mathrm{area}(\Omega)}{\mathrm{area}(\mathrm{conv}(\Omega)).}
  \end{align*}
  
  Since the intersection of convex sets is a convex set, there is a unique smallest (by containment) convex hull for any region $\Omega$.
\end{definition}

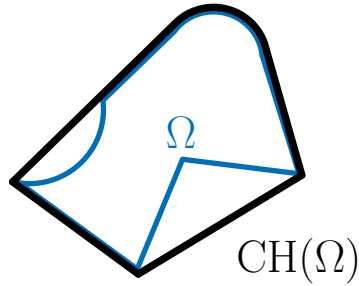
\begin{figure}[htb]
	\centering
	\definecolor{qqqqff}{rgb}{0,0,1}
\begin{tikzpicture}[scale=1,line cap=round,line join=round,>=triangle 45,x=1cm,y=1cm]
\clip(-3.1808855374298814,0.7054196919809391) rectangle (2.8436799719688346,4.78259146738453);
\draw [line width=2pt,color=NavyBlue] (-0.5403387932320529,4.199917026631047)-- (-1.4876311065851997,3.283901898015972);
\draw [line width=2pt,color=NavyBlue] (-2.59555742277687,2.1884264530340003)-- (-1.030093019846988,1.0427513872026393);
\draw [line width=2pt,color=NavyBlue] (-1.030093019846988,1.0427513872026393)-- (-0.4266458477678716,2.503268455857886);
\draw [line width=2pt,color=NavyBlue] (-0.4266458477678716,2.503268455857886)-- (1.1213273327829054,2.310865009687734);
\draw [line width=2pt,color=NavyBlue] (1.1213273327829054,2.310865009687734)-- (0.6315731061679702,3.8850750238071616);
\draw [shift={(-2.4623421270336228,3.1638741521373337)},line width=2pt,color=NavyBlue]  plot[domain=-1.7065250247888804:0.12252504290266225,variable=\t]({1*0.9845021730326186*cos(\t r)+0*0.9845021730326186*sin(\t r)},{0*0.9845021730326186*cos(\t r)+1*0.9845021730326186*sin(\t r)});
\draw [shift={(-0.01568921678809781,3.814254744830714)},line width=2pt,color=NavyBlue]  plot[domain=0.10898159374480773:2.5077053734556625,variable=\t]({1*0.6511252004282949*cos(\t r)+0*0.6511252004282949*sin(\t r)},{0*0.6511252004282949*cos(\t r)+1*0.6511252004282949*sin(\t r)});
\draw [line width=2.8pt] (-2.685655712545722,2.2031756923666004)-- (-1.0153159792750968,0.9768503185729869);
\draw [line width=2.8pt] (-1.0153159792750968,0.9768503185729869)-- (1.158637184087576,2.29171136125267);
\draw [line width=2.8pt] (1.158637184087576,2.29171136125267)-- (0.6620716009713541,4.04266375305702);
\draw [shift={(-0.01540854579167931,3.8076224532316205)},line width=2.8pt]  plot[domain=0.33394133978769747:2.3774678762164125,variable=\t]({1*0.7170939700497241*cos(\t r)+0*0.7170939700497241*sin(\t r)},{0*0.7170939700497241*cos(\t r)+1*0.7170939700497241*sin(\t r)});
\draw [line width=2.8pt] (-2.685655712545722,2.2031756923666004)-- (-0.5331419341269765,4.30378361705294);
\draw [color=NavyBlue](-0.7863560820341027,3.2000351515748306) node[anchor=north west] {\LARGE$\Omega$};
\draw (0.1549822788094467,1.5938421717062845) node[anchor=north west] {\LARGE$\mathrm{CH}(\Omega{)}$};
\end{tikzpicture}
	\caption{A region $\Omega$ and its convex hull.}
	\label{fig:ch_example}
\end{figure}

Suppose that our map projection $\vphi$ does  preserve the ordering of regions induced by the convex hull score.  We begin by observing that such a projection must preserve certain geometric properties of regions within its domain.
\begin{lemma}~\label{lem:CH_prep}
	Let $\vphi$ be a map projection from some region of the sphere to a region of the plane. If $\vphi$ preserves the convex hull compactness score ordering, then the following must 
	hold:
	\begin{enumerate}
		\item $\vphi$ and $\vphi^{-1}$ send convex regions in their domains to convex regions in their codomains.
		\item $\vphi$ sends every segment of a great circle in its domain to a line segment in its codomain.  That is, it preserves geodesics.\footnote{Such a projection is sometimes called a \textit{geodesic map}.}
		\item There exists a region $U$ in the domain of $\vphi$
		such that for any regions $A,B\subset U$, if 
		$A$ and $B$ have equal area on the sphere, then 
		$\vphi(A)$ and $\vphi(B)$ have equal area in the plane.  The same holds 
		for $\vphi^{-1}$ for all pairs of regions inside of $\vphi(U)$.
	\end{enumerate}
\end{lemma}

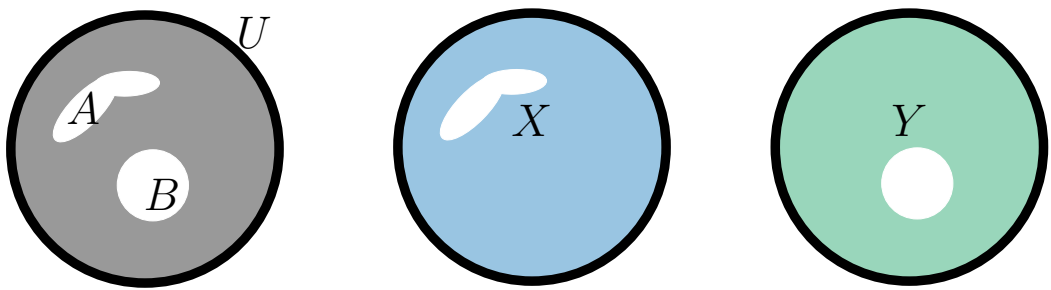
\begin{figure}[h]
	\centering
	\begin{minipage}{.3\textwidth}
			\centering
	\definecolor{ffffff}{rgb}{1,1,1}

\begin{tikzpicture}[scale=.35,line cap=round,line join=round,>=triangle 45,x=1cm,y=1cm]

\draw [line width=3.6pt,fill=black,fill opacity=0.4] (2.52,-1.24) circle (5.094349811310566cm);
\draw [line width=0.4pt,color=ffffff,fill=ffffff,fill opacity=1] (2.82,-2.62) circle (1.352035502492446cm);
\draw [rotate around={-134.1620339448981:(0.25635345523487,0.2782885419108182)},line width=0.4pt,color=ffffff,fill=ffffff,fill opacity=1] (0.25635345523487,0.2782885419108182) ellipse (1.6381871321999442cm and 0.5964014274522406cm);
\draw [rotate around={2.00955381302114:(1.84,1.24)},line width=0.4pt,color=ffffff,fill=ffffff,fill opacity=1] (1.84,1.24) ellipse (1.2329931677283157cm and 0.46805144125908144cm);
\draw (-0.84942211242749375,1.2652523257265278) node[anchor=north west] {\LARGE$A$};
\draw (2.089030655746763,-1.9839442053140297) node[anchor=north west] {\LARGE$B$};
\draw (5.613994458370557,4.161101819701755) node[anchor=north west] {\LARGE$U$};
\end{tikzpicture}
			\end{minipage}
	\begin{minipage}{.3\textwidth}
			\centering
	\definecolor{ffffff}{rgb}{1,1,1}

\begin{tikzpicture}[scale=.35,line cap=round,line join=round,>=triangle 45,x=1cm,y=1cm]

\draw [line width=3.6pt,fill=NavyBlue,fill opacity=0.4] (2.52,-1.24) circle (5.094349811310566cm);
\draw [rotate around={-134.1620339448981:(0.25635345523487,0.2782885419108182)},line width=0.4pt,color=ffffff,fill=ffffff,fill opacity=1] (0.25635345523487,0.2782885419108182) ellipse (1.6381871321999442cm and 0.5964014274522406cm);
\draw [rotate around={2.00955381302114:(1.84,1.24)},line width=0.4pt,color=ffffff,fill=ffffff,fill opacity=1] (1.84,1.24) ellipse (1.2329931677283157cm and 0.46805144125908144cm);
\draw (2.52,-0.24) node[anchor=center] {\LARGE$X$};
\end{tikzpicture}
\end{minipage}
	\begin{minipage}{.3\textwidth}
			\centering
	\definecolor{ffffff}{rgb}{1,1,1}

\begin{tikzpicture}[scale=.35,line cap=round,line join=round,>=triangle 45,x=1cm,y=1cm]

\draw [line width=3.6pt,fill=ForestGreen,fill opacity=0.4] (2.52,-1.24) circle (5.094349811310566cm);
\draw [line width=0.4pt,color=ffffff,fill=ffffff,fill opacity=1] (2.82,-2.62) circle (1.352035502492446cm);

\draw (2.52,-0.24) node[anchor=center] {\LARGE$Y$};
\end{tikzpicture}
\end{minipage}
	\caption{Two equal area regions $A$ and $B$ removed from $U$ to form the regions $X$ and $Y$.}
	\label{fig:ch_schema}
\end{figure}

\begin{proof}

		The proof of (1) follows from the idea that any projection which preserves the convex hull score ordering of regions must 
		preserve the maximizers in that ordering.   Here, the maximizers are convex sets.
		
		 To show (2) we suppose, for the sake of contradiction, that there is some geodesic segment $s$ in $U$ such that $\vphi(s)$ is not a line segment. Construct two convex spherical polygons $L$ and $M$ inside of $U$ which both have $s$ as a side.

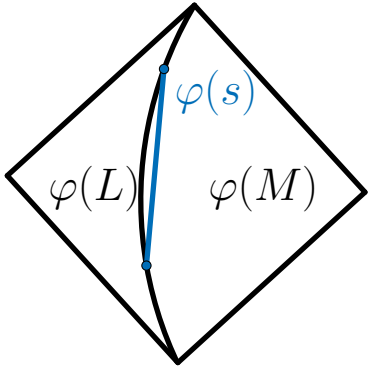
\begin{figure}[h]
	\centering
	\begin{tikzpicture}[scale=.7,line cap=round,line join=round,>=triangle 45,x=1cm,y=1cm]

\clip(-9.2605723047488,4.939165952768963) rectangle (5.35451137900598,13.039739619301336);
\draw [line width=2pt] (-2.0244628099173525,12.456528925619821)-- (-5.578181818181814,9.216859504132223);
\draw [line width=2pt] (-5.578181818181814,9.216859504132223)-- (-2.338512396694217,5.663140495867761);
\draw [line width=2pt] (-2.338512396694217,5.663140495867761)-- (1.2152066115702453,8.90280991735536);
\draw [line width=2pt] (1.2152066115702453,8.90280991735536)-- (-2.0244628099173525,12.456528925619821);
\draw [shift={(4.190413223140495,8.754049586776851)},line width=2pt]  plot[domain=2.604307939289022:3.5793780144036975,variable=\t]({1*7.234157681502106*cos(\t r)+0*7.234157681502106*sin(\t r)},{0*7.234157681502106*cos(\t r)+1*7.234157681502106*sin(\t r)});
\draw [line width=2pt,color=NavyBlue] (-2.601747508322699,11.243947166759868)-- (-2.9355456436406637,7.507555467344712);
\draw (-4.988752951588037,9.58982716176266) node[anchor=north west] {\LARGE$\varphi(L)$};
\draw (-2.588752951588037,11.42982716176266) node[color=NavyBlue,anchor=north west] {\LARGE$\varphi(s)$};

\draw (-1.9481756813548882,9.58982716176266) node[anchor=north west] {\LARGE$\varphi(M)$};
\begin{scriptsize}
\draw [fill=NavyBlue] (-2.601747508322699,11.243947166759868) circle (2.5pt);
\draw [fill=NavyBlue] (-2.9355456436406637,7.507555467344712) circle (2.5pt);
\end{scriptsize}
\end{tikzpicture}
	\caption{If $\vphi(s)$ is not a line segment, then one of $\vphi(M)$ or $\vphi(L)$ is not convex.}
	\label{fig:lineconvexcont}
\end{figure}

		 By (1), $\vphi$ must send both of these polygons to convex regions in the plane, but this is not the case.  All of the points along $\vphi(s)$ belong to both $\vphi(L)$ and $\vphi(M)$, but since $\vphi(s)$ is not a line segment, we can find two points along it which are joined by some line segment which contains points which only belong to $\vphi(L)$ or $\vphi(M)$, which means that at least one of these convex spherical polygons is sent to something non-convex in the plane, which contradicts our assumption.		See \Cref{fig:lineconvexcont} for an illustration.
		
		That $\vphi^{-1}$ sends line segments in the plane to great circle segments on the sphere is shown analogously.  
		This completes the proof of (2).

 To show (3), let $U$ be some convex region in the domain of $\vphi$.  Take $A,B$ to be regions of equal area such that $A$ and $B$ are properly contained in the interior of $U$, as in \Cref{fig:ch_schema}.  Define two new regions $X=U\ssm A$ and $Y=U\ssm B$, i.e. these regions are equal to $U$ with $A$ or $B$ deleted, respectively.  

The cap $U$ is itself the convex hull of both $X$ and $Y$, and since $A$ and $B$ have equal area, we have that $\mathrm{CH}(X) = \mathrm{CH}(Y)$.  Since $U$ is a cap, it is convex, so by (1), $\vphi(U)$ is also convex.  Since $\vphi$ preserves the ordering of convex hull scores and $X$ and $Y$ had equal scores on the sphere, $\vphi$ must send $X$ and $Y$ to regions in the plane which also have the same convex hull score as each other.  Furthermore, the convex hulls of $\vphi(X)$ and $\vphi(Y)$ are $\vphi(U)$.

By definition, we have
\begin{align*}
\mathrm{CH}(X) &= \mathrm{CH}(Y)\\
\end{align*}
and by the construction of $X$ and $Y$, we have 
\begin{align*}
\frac{\mathrm{area}(\vphi(U)) - \mathrm{area}(\vphi(A))}{\mathrm{area}(\vphi(U))} &= \frac{\mathrm{area}(\vphi(U)) - \mathrm{area}(\vphi(B))}{\mathrm{area}(\vphi(U))}\\
\mathrm{area}(\vphi(A)) &= \mathrm{area}(\vphi(B))
\end{align*}
 which is what we wanted to show.  The proof that $\vphi^{-1}$ also has this property is analogous.

\end{proof}

We can now show that no map projection can preserve the convex hull score ordering of regions by demonstrating that there is no projection from a patch on the sphere to the plane which has all three of the properties described  in \Cref{lem:CH_prep}.

\begin{theorem}
	There does not exist a map projection with the three properties in Lemma~\ref{lem:CH_prep}
\end{theorem}
\begin{proof}
	Assume that such a map, $\vphi$, exists, and restrict 
	it to $U$ as above. Let $T\subset U$ be a 
	small equilateral spherical  triangle centered at 
	the center of $U$. Let $T_1$ and $T_2$ be two 
	congruent triangles meeting at a point and 
	each sharing a face with $T$, as in \Cref{fig:sphtris}.

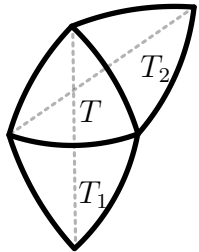
\begin{figure}[!htb]
	\centering
	\begin{tikzpicture}[scale=.25,line cap=round,line join=round,>=triangle 45,x=1cm,y=1cm]
\clip(-7.699047690559778,-7.571555155097503) rectangle (6.4348506610533125,7.774557106234026);
\draw [shift={(-1.6505275214910575,9.115278552365504)},line width=2pt]  plot[domain=4.3817706171001:5.059138387806136,variable=\t]({1*10.321001528275485*cos(\t r)+0*10.321001528275485*sin(\t r)},{0*10.321001528275485*cos(\t r)+1*10.321001528275485*sin(\t r)});
\draw [shift={(-10.803070399796216,-4.420666087580098)},line width=2pt]  plot[domain=0.2938195101807776:0.8077763069010403,variable=\t]({1*13.226440769305302*cos(\t r)+0*13.226440769305302*sin(\t r)},{0*13.226440769305302*cos(\t r)+1*13.226440769305302*sin(\t r)});
\draw [shift={(8.343628495049057,-4.490800515766344)},line width=2pt]  plot[domain=2.3753688206179033:2.8611250423804107,variable=\t]({1*13.887257489360788*cos(\t r)+0*13.887257489360788*sin(\t r)},{0*13.887257489360788*cos(\t r)+1*13.887257489360788*sin(\t r)});
\draw [shift={(-7.513162654544307,7.579158955824589)},line width=2pt]  plot[domain=5.566336287533996:6.169114070716076,variable=\t]({1*12.436992398719607*cos(\t r)+0*12.436992398719607*sin(\t r)},{0*12.436992398719607*cos(\t r)+1*12.436992398719607*sin(\t r)});
\draw [shift={(3.7978267211845975,-8.419128622621312)},line width=2pt]  plot[domain=1.4994026026391953:1.954080474540578,variable=\t]({1*14.614520953612525*cos(\t r)+0*14.614520953612525*sin(\t r)},{0*14.614520953612525*cos(\t r)+1*14.614520953612525*sin(\t r)});
\draw [shift={(10.390041391475169,4.172878076060955)},line width=2pt]  plot[domain=3.4450650540010272:3.881445927448308,variable=\t]({1*16.127936863518837*cos(\t r)+0*16.127936863518837*sin(\t r)},{0*16.127936863518837*cos(\t r)+1*16.127936863518837*sin(\t r)});
\draw [shift={(-9.037074369684714,1.4207160457879888)},line width=2pt]  plot[domain=5.463306638850501:6.100385463261517,variable=\t]({1*11.066136404446647*cos(\t r)+0*11.066136404446647*sin(\t r)},{0*11.066136404446647*cos(\t r)+1*11.066136404446647*sin(\t r)});
\draw [line width=1.2pt,dotted,opacity=.3] (4.8403256579539065,6.1581625753487135)-- (-5.001005706626561,-0.6467585545910133);
\draw [line width=1.2pt,dotted,opacity=.3] (-1.613022483529527,5.156840351541236)-- (-1.4587194302885766,-6.643272504200754);
\draw (-1.7868492266322078,1.7623586423065475) node[anchor=north west] {\large$T$};
\draw (-1.8511013622691343,-2.593740589657192) node[anchor=north west] {\large$T_1$};
\draw (1.458165010947301,4.1125782470569596) node[anchor=north west] {\large$T_2$};
\end{tikzpicture}
	\caption{The spherical regions $T,T_1,T_2$.}
	\label{fig:sphtris}
\end{figure}

We first argue that the images of $T\cup T_1$ and $T\cup T_2$ are parallelograms.

Without loss of generality, consider $T\cup T_1$.  By construction, it is a 
convex spherical quadrilateral. By symmetry, its geodesic 
diagonals on the sphere divide it into four triangles of equal area.  To see this, consider the geodesic segment which passes through the vertex of $T$ opposite the side shared with $T_1$ which divides $T$ into two smaller triangles of equal area.  Since $T$ is equilateral, this segment meets the shared side at a right angle at the midpoint, and the same is true for the area bisector of $T_1$.  Since both of these bisectors meet the shared side at a right angle and at the same point, together they form a single geodesic segment, the diagonal of the quadrilateral.  Since the diagonal cuts each of $T$ and $T_1$ in half, and $T$ and $T_1$ have the same area, the four triangles formed in this construction have the same area.

		Since $\vphi$ sends spherical geodesics to line segments in the plane, it must send 
		$T\cup T_1$ to a Euclidean quadrilateral $Q$ whose diagonals 
		are the images of the diagonals of the spherical quadrilateral $T\cup T_1$.
		
		 Since 
		$\vphi$ sends equal area regions to equal area 
		regions, it follows that the diagonals 
		of $Q$ split it into four equal area triangles.
		
		We now argue that this implies that $Q$ is a Euclidean parallelogram by showing that its diagonals bisect each other.  Since the four triangles 
		formed by the diagonals of $Q$ are all the same area, we can pick two of these triangles which share a side 
		and consider the larger triangle formed by their union.  One side of this triangle is a diagonal $d_1$ of $Q$ and its area is 
		bisected by the other diagonal $d_2$, which passes through $d_1$ and its opposite vertex.  The area bisector from a vertex, called the \textit{median}, passes through the midpoint of the side $d_1$, meaning that the diagonal $d_2$ bisects the diagonal $d_1$.  Since this holds for any choice of two adjacent triangles in $Q$, the diagonals must bisect each other, so $Q$ is a parallelogram.
		
		\begin{figure}[h]
			\centering
			\begin{tikzpicture}[scale=.6,line cap=round,line join=round,>=triangle 45,x=1cm,y=1cm]
\clip(-0.5194903865123804,-2.5120492887047456) rectangle (6.611417807606881,3.221586353864317);
\fill[opacity=0,line width=2pt] (2.7418465885478605,2.9045288146832235) -- (0.4115999626018961,0.44880737041701263) -- (1.56,-2.08) -- (5.58,-2.08) -- cycle;
\draw [line width=2pt] (2.7418465885478605,2.9045288146832235)-- (0.4115999626018961,0.44880737041701263);
\draw [line width=2pt] (0.4115999626018961,0.44880737041701263)-- (1.56,-2.08);
\draw [line width=2pt] (1.56,-2.08)-- (5.58,-2.08);
\draw [line width=2pt] (5.58,-2.08)-- (2.7418465885478605,2.9045288146832235);
\draw [line width=2pt] (4.163979603645924,0.4068967404389632)-- (0.4115999626018961,0.44880737041701263);
\draw [line width=2pt] (4.163979603645924,0.4068967404389632)-- (1.56,-2.08);
\draw [line width=1pt,dotted,opacity=.3] (2.7418465885478605,2.9045288146832235)-- (1.56,-2.08);
\draw [line width=1pt,dotted,opacity=.3] (5.58,-2.08)-- (0.4115999626018961,0.44880737041701263);
-\draw (2.978509173575852,-.8725949117012167) node[anchor=north west] {\large$\varphi(T_1)$};
-\draw (5.078509173575852,-.8725949117012167) node[anchor=north west] {\large$m_1$};
\draw (1.706327414776063,1.5774047550323834) node[anchor=north west] {\large$\varphi(T_2)$};
\draw (3.706327414776063,1.5774047550323834) node[anchor=north west] {\large$m_2$};
\draw (1.2021456159842953,-0.08849580106548587) node[anchor=north west] {\large$\varphi(T)$};
\draw (0.1021456159842953,-0.38849580106548587) node[anchor=north west] {\large$\ell$};
\end{tikzpicture}
			\caption{The image under $\vphi$ of $T,T_1,T_2$ which form the quadrilateral in the plane.}
			\label{fig:sphtris_pl}
		\end{figure}
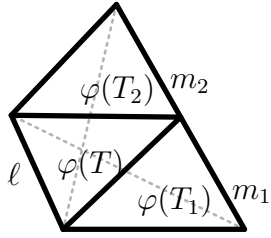

Since $T\cup T_1$ and $T\cup T_2$ are both spherical quadrilaterals which overlap on the spherical triangle $T$, the images of $T\cup T_1$ and $T\cup T_2$ are Euclidean parallelograms of equal area which overlap on a shared triangle $\vphi(T)$. See \Cref{fig:sphtris_pl} for an illustration.
  
Because  the segment $\ell$ is parallel to $m_1$ and $m_2$, $m_1$ and $m_2$ are parallel to each other, and because they meet at the point shared by all three triangles, $m_1$ and $m_2$ together form a single segment parallel to $\ell$.  Therefore, the image of the three triangles forms a quadrilateral in the plane.
	Therefore, the image of $T\cup T_1\cup T_2$ has a boundary consisting of 
	four line segments.
	
	To find the contradiction, consider the point on the sphere shared by $T$, $T_1$, and $T_2$.  Since these triangles are all equilateral spherical triangles, the three angles at this point are each strictly greater than $\tfrac{\pi}{3}$ radians, because the sum of interior angles on a triangle is strictly greater than $\pi$.  
	so, the total measure of the three angles at this point is greater than $\pi$,  Therefore, the two geodesic segments which are part of the boundaries of $T_1$ and $T_2$ meet at this point at an angle of measure strictly greater than 
	$\pi$. Therefore, together they do not form a single geodesic.  On the sphere, the region $T\cup T_1\cup T_2$ has a boundary consisting of five geodesic segments whereas its image has a boundary consisting of four, which contradicts the assumption that $\vphi$ and $\vphi^{-1}$ preserve geodesics.
\end{proof}

This implies that no map projection can preserve the ordering of regions by their convex hull scores, which is what we aimed to show.

\section{Reock}\label{sec:reock}

Let $\mathrm{circ}(\Omega)$ denote the \textit{smallest bounding
circle} (smallest bounding \textit{cap} on the sphere) of a region
$\Omega$.  Then the \textit{Reock score} of $\Omega$ is 

$$\mathrm{Reock}(\Omega)=
\frac{\mathrm{area}(\Omega)}{\mathrm{area}(\mathrm{circ}(\Omega))}.$$

We again consider what properties a map projection $\vphi$ must have in order to preserve the ordering of regions by their Reock scores.  

\begin{lemma}\label{lem:reock_prep}
  If $\vphi$ preserves the ordering of regions induced by their Reock scores, then the following must hold:
  \begin{enumerate}
    \item $\vphi$ sends spherical caps in its domain to Euclidean circles in the plane,  and $\vphi^{-1}$ does the opposite. 
    \item There exists a region $U$ in the domain of $\vphi$ such that for any regions $A,B\subset U$, if $A$ and $B$ have equal area on the sphere, then $\vphi(A)$ and $\vphi(B)$ have equal area in the plane.  The same holds for $\vphi^{-1}$ for all pairs of regions inside of $\vphi(U)$.
  \end{enumerate}
\end{lemma}
\begin{proof}
  Similarly to the convex hull setting, the proof of (1) follows from
  the requirement that $\vphi$ preserves the maximizers in the
  compactness score ordering.  In the case of the Reock score, the
  maximizers are caps in the sphere and circles in the plane.  

  To show (2), let $\kappa$ be a cap in the domain of $\vphi$, and let 
  $A,B\subset \kappa$ be two regions of equal area properly 
  contained in the interior of $\kappa$. Then, define two new regions
  $X=\kappa\ssm A$ and $Y=\kappa\ssm B$, which can be thought of as
  $\kappa$ with $A$ and $B$ deleted, respectively. 

  Since $\kappa$ is the smallest bounding cap of $X$ and $Y$ and since
  $A$ and $B$ have equal areas, $\mathrm{Reock}(X)=\mathrm{Reock}(Y)$.
  Furthermore, by (1), $\vphi$ must send $\kappa$ to some circle in
  the plane, which is the smallest bounding circle of $\vphi(X)$ and
  $\vphi(Y)$.    
  Since $\vphi$ preserves the ordering of Reock scores, it must be
  that $\vphi(X)$ and $\vphi(Y)$ have identical Reock scores in the
  plane.

  By definition, we can write
  \begin{align*}
    \mathrm{Reock}(X) &= \mathrm{Reock}(Y)\\
    \frac{\mathrm{area}(\vphi(X))}{\mathrm{area}(\vphi(\kappa))} &= \frac{\mathrm{area}(\vphi(Y))}{\mathrm{area}(\vphi(\kappa))}
  \end{align*}
  and by the construction of $X$ and $Y$, we have 
  \begin{align*}
    \frac{\mathrm{area}(\vphi(\kappa))
    - \mathrm{area}(\vphi(A))}{\mathrm{area}(\vphi(\kappa))} &=
    \frac{\mathrm{area}(\vphi(\kappa))
    - \mathrm{area}(\vphi(B))}{\mathrm{area}(\vphi(\kappa))}\\
    \mathrm{area}(\vphi(A)) &= \mathrm{area}(\vphi(B)),
  \end{align*}
  \mute{  
  \begin{align*}
    \mathrm{Reock}(X) 
    = 1-\frac{\mathrm{area}(A)}{\mathrm{area}(\kappa)} 
    = 1-\frac{\mathrm{area}(B)}{\mathrm{area}(\kappa)}
    = \mathrm{Reock}(Y)
  \end{align*}
  Note that $\kappa$ is sent to a circle in the plane, so 
  the minimal bounding circle of $\vphi(\kappa)$ is itself. Thus, 
  since $\vphi$ preserves equality of Reock scores,
  \begin{align*}
    1-\frac{\mathrm{area}(\vphi(A))}{\mathrm{area}(\vphi(\kappa))} 
    = \mathrm{Reock}(\vphi(X)) 
    = \mathrm{Reock}(Y)
    = 1-\frac{\mathrm{area}(\vphi(B))}{\mathrm{area}(\vphi(\kappa))}
  \end{align*}

  }
  meaning that $\mathrm{area}(\vphi(A)) = \mathrm{area}(\vphi(B))$. 
  Thus, for all pairs of regions of the same area inside of $\kappa$,
  the images under $\vphi$ of those regions will have the same area as
  well.

  The same construction works in reverse, which demonstrates that
  $\vphi^{-1}$ also sends regions of equal area in some circle in the
  plane to regions of  equal area in the sphere.
\end{proof} 

We can now show that no such $\vphi$ exists.  Rather than constructing a figure on the sphere and examining its image under $\vphi$, it will be more convenient to construct a figure in the plane and reason about $\vphi^{-1}$.

\begin{theorem}\label{thm:reockbad}
  There does not exist a map projection with the two properties in
  \Cref{lem:reock_prep}.  
\end{theorem}
\begin{proof}

  Assume that such a $\vphi$ does exist and restrict its domain to
  a cap $\kappa$ as above.  This corresponds to a restriction of the
  domain of $\vphi^{-1}$ to a circle in the plane.  Inside of this
  circle, draw seven smaller circles of equal area tangent to each
  other as in \Cref{fig:sevencircles}.

  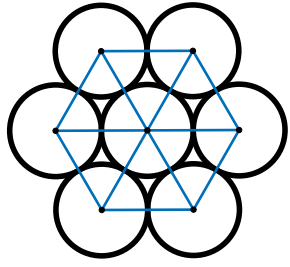
\begin{figure}[!htb]
    
    \centering
    \begin{tikzpicture}[scale=.25,line cap=round,line join=round,>=triangle 45,x=1cm,y=1cm]
\clip(-11.346776859504123,1.0846280991735546) rectangle (9.810247933884293,16.026776859504118);
\draw [line width=2.1pt] (-2.6746410161513783,12.872407940936126) circle (2.4400204917172323cm);
\draw [line width=2.1pt] (-5.097320508075692,8.636203970468065) circle (2.4400204917172315cm);
\draw [line width=2.1pt] (-2.64,4.42) circle (2.440020491717234cm);
\draw [line width=2.1pt] (2.24,4.44) circle (2.4400204917172372cm);
\draw [line width=2.1pt] (4.662679491924313,8.676203970468062) circle (2.440020491717239cm);
\draw [line width=2.1pt] (2.2053589838486225,12.892407940936126) circle (2.4400204917172372cm);
\draw [line width=2.1pt] (-0.21732050807568942,8.656203970468063) circle (2.4435284258377075cm);
\draw [line width=1pt,color=NavyBlue] (-2.6746410161513783,12.872407940936126)-- (-0.21732050807568942,8.656203970468063);
\draw [line width=1pt,color=NavyBlue] (-0.21732050807568942,8.656203970468063)-- (2.2053589838486225,12.892407940936126);
\draw [line width=1pt,color=NavyBlue] (2.2053589838486225,12.892407940936126)-- (-2.6746410161513783,12.872407940936126);
\draw [line width=1pt,color=NavyBlue] (-2.6746410161513783,12.872407940936126)-- (-5.097320508075692,8.636203970468065);
\draw [line width=1pt,color=NavyBlue] (-5.097320508075692,8.636203970468065)-- (-2.64,4.42);
\draw [line width=1pt,color=NavyBlue] (-2.64,4.42)-- (2.24,4.44);
\draw [line width=1pt,color=NavyBlue] (2.24,4.44)-- (4.662679491924313,8.676203970468062);
\draw [line width=1pt,color=NavyBlue] (4.662679491924313,8.676203970468062)-- (2.2053589838486225,12.892407940936126);
\draw [line width=1pt,color=NavyBlue] (4.662679491924313,8.676203970468062)-- (-0.21732050807568942,8.656203970468063);
\draw [line width=1pt,color=NavyBlue] (-0.21732050807568942,8.656203970468063)-- (2.24,4.44);
\draw [line width=1pt,color=NavyBlue] (-0.21732050807568942,8.656203970468063)-- (-2.64,4.42);
\draw [line width=1pt,color=NavyBlue] (-0.21732050807568942,8.656203970468063)-- (-5.097320508075692,8.636203970468065);
\begin{scriptsize}
\draw [fill=black] (-2.64,4.42) circle (4pt);
\draw [fill=black] (2.24,4.44) circle (4pt);
\draw [fill=black] (4.662679491924313,8.676203970468062) circle (4pt);
\draw [fill=black] (2.2053589838486225,12.892407940936126) circle (4pt);
\draw [fill=black] (-2.6746410161513783,12.872407940936126) circle (4pt);
\draw [fill=black] (-5.097320508075692,8.636203970468065) circle (4pt);
\draw [fill=black] (-0.21732050807568942,8.656203970468063) circle (4pt);
\end{scriptsize}
\end{tikzpicture}
    \caption{Seven circles arranged as in the construction for \Cref{thm:reockbad}.}
    \label{fig:sevencircles}
  \end{figure}	

  Under $\vphi^{-1}$, they must be sent to an similar configuration 
  of equal-area caps on the sphere .  

  However, the radius of a
  of a spherical cap is determined by its area, so since the areas of
  these caps are all the same, their radii must be as well. Thus, 
  the midpoints of these caps form six equilateral triangles on the sphere
  which meet at a point.  However, this is impossible, as the three 
  angles of an equilateral triangle on the sphere must all be greater
  than $\tfrac{\pi}{3}$, but the total measure of all the angles at
  a point must be equal to
  $2\pi$, which contradicts the assumption that such a $\vphi$ exists.
\end{proof}

This shows that no map projection exists which preserves the ordering of regions by their Reock scores.

\section{Polsby-Popper}\label{sec:pp}
The final compactness score we analyze is the \textit{Polsby-Popper
score}, which takes the form of an \textit{isoperimetric quotient},
meaning it measures how much area a region's perimeter encloses,
relative to all other regions with the same perimeter.

\begin{definition}\label{def:pp}
  The Polsby-Popper score of a region $\Omega$ is defined to be
  $$\mathrm{PP}(\Omega) = \frac{4\pi
  \cdot\mathrm{area}(\Omega)}{\mathrm{perim}(\Omega)^2}$$ 
in either the sphere or the plane, and
  $\mathrm{area}$ and $\mathrm{perim}$ are the area and perimeter of
    $\Omega$, respectively.
\end{definition}

The ancient Greeks were first to observe that if $\Omega$ is a region
in the plane, then $4\pi\cdot\mathrm{area}(\Omega)\leq
\mathrm{perim}(\Omega)^2$, with equality if and only if $\Omega$ is
a circle. This became known as the \textit{isoperimetric inequality} in
the plane.  This means that, in the plane, $0\le \mathrm{PP}(\Omega)\le 1$,
where the Polsby-Popper score is equal to $1$ only in the case of
a circle. We can observe that the Polsby-Popper score is scale-invariant in
the plane. 

An isoperimetric inequality for the sphere exists, and we
state it as the following lemma.  For a more detailed treatment of
isoperimetry in general, see \cite{osserman1979bonnesen}, and for
a proof of this inequality for the sphere, see \cite{rado}.

\begin{lemma}
  If $\Omega$ is a region on the sphere with area
  $A$ and perimeter $P$, then $P^2\geq 4\pi A - A^2$ with equality if
  and only if $\Omega$ is a spherical cap.
\end{lemma}
A consequence of this is that among all regions on the sphere with
a fixed area $A$, a spherical cap with area $A$ has the shortest
perimeter. However, the key difference between the Polsby-Popper score in the plane and on the sphere is that on the 
sphere, there is no scale invariance; two spherical caps of different sizes will have different scores.

\begin{lemma}\label{lem:ppscale}
  Let $S$ be the unit sphere, and let $\kappa(h)$ be a cap of height
  $h$.  Then $\mathrm{PP}(\kappa(h))$ is
  a monotonically increasing function of $h$.
\end{lemma}

\begin{proof}
  Let $r(h)$ be the radius of the circle bounding $\kappa(h)$. We
  compute: 
  \begin{align*}
    1 &= r(h)^2 + (1-h)^2 \text {, by right triangle trigonometry}\\ 
      &= r(h)^2 + 1 - 2h+h^2
  \end{align*}
  Rearranging, we get that $r(h)^2= 2h-h^2$, which we can plug in to
  the standard formula for perimeter:
  \begin{align*}
    \mathrm{perim}_S(\kappa(h)) = 2\pi r(h) = 2\pi \sqrt{2h-h^2}
  \end{align*}
  We can now use the Archimedian equal-area projection 
  defined by $(x,y,z) \to
  \left(\frac{x}{\sqrt{x^2+y^2}},\frac{y}{\sqrt{x^2+y^2}}, z\right)$ 
  to compute $\mathrm{area}_S(\kappa(h)) = 2\pi h$ and plug it in to 
  get:

  \begin{align*}
    \mathrm{PP}_S(\kappa(h)) = \frac{4\pi (2\pi h) }{4 \pi^2 (2h-h^2)}
    = \frac{2}{2-h}
  \end{align*}
  Which is a monotonically increasing function of $h$.
\end{proof}
\begin{corollary}\label{cor:capscale}
  On the sphere, Polsby-Popper scores of caps are monotonically
  increasing with area,
\end{corollary}
Using this, we can show the main theorem of this section, that no map
projection from a region on the sphere to the plane can preserve the ordering
of Polsby-Popper scores for all regions.  

\begin{theorem}\label{thm:dpp}
  If $\varphi:U\to V$ is a map projection from the sphere to the plane,
  then there exist two regions $A,B\subset U$ such that
  the Polsby-Popper score of $B$ is greater than that of $A$ in the
  sphere, but the Polsby-Popper score of $\varphi(A)$ is greater than
  that of $\varphi(B)$ in the plane.
\end{theorem}
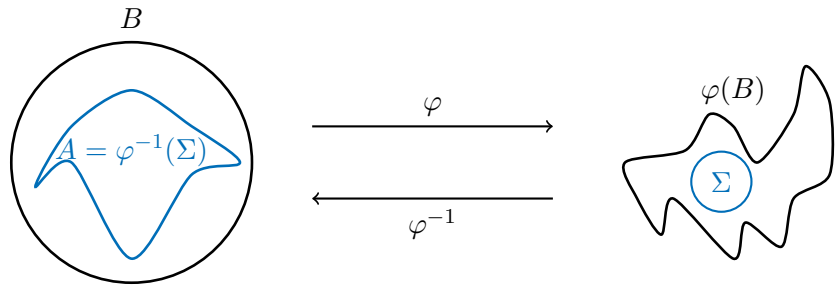
\begin{figure}[h]
  \centering
  \begin{tikzpicture}[scale=1.6]
    \draw[thick,->] (-1,0.3)-- (1,0.3)%
    node[midway,above] {$\vphi$};
    \draw[thick,->] (1,-0.3) -- (-1,-0.3)%
    node[midway,below] {$\vphi^{-1}$}; 
    \draw[line width=1, black] (-2.5,0) circle (1);

    \draw node at (-2.5,1.2) {\color{black} $B$};
    \draw[line width=1, black] plot [smooth cycle, tension=0.5]%
    coordinates { (2.5,-0.8) (2.6,-0.5) (2.9,-0.7) (3,-0.3) (3.3,-0.1)% 
    (3.3,0.5) (3.1,0.8) (3,0.4) (2.7,0) (2.5,0.3) (2.3,0.4)%
    (2.1,0.1) (1.6,0) (1.7,-0.3) (1.9, -0.6) (2,-0.3)};
    \draw node at (2.5,0.6) {\color{black} $\vphi(B)$};

    \draw[thick, NavyBlue] (2.4,-0.16) circle (0.25);
    \draw node at (2.4,-0.16) {\color{NavyBlue} $\Sigma$};
    \draw[line width=1, NavyBlue] plot [smooth cycle, tension=0.5]%
    coordinates { (-2.5,-0.8) (-3, 0) (-3.3,-0.2) (-3,0.3) (-2.5,0.6)%
    (-2,0.3) (-1.6,0) (-2,-0.1)};
    \draw node at (-2.5,0.1) {\color{NavyBlue} $A=\vphi^{-1}(\Sigma)$};

  \end{tikzpicture}
  
  \caption{The construction of regions $A$ and $B$ in the
  proof of Theorem~\ref{thm:dpp}.} 
\label{fig:dpp}
\end{figure}

\begin{proof}
  Let $\varphi$ be a map projection, and let 
  $\kappa \subset U$ be some cap. We will take our regions 
  $A$ and $B$ to lie in $\kappa$. Set $B$ to be a cap 
  contained in $\kappa$. Let $\Sigma$ be a circle in 
  the plane such that $\Sigma
  \subsetneq \varphi(B)$ and let $A=\varphi^{-1}(\Sigma)$. See
  Figure~\ref{fig:dpp} for an illustration.

  We now use the isoperimetric inequality for the sphere 
  and Corollary~\ref{cor:capscale} to claim that 
  $A$ does not maximize the Polsby-Popper score in the sphere.

  To see this, take $\hat{A}$ to be a cap in the sphere with 
  area equal to that of $A$. Note that since the 
  area of $\hat {A}$ is less than the area of the cap $B$, it 
  follows that we can choose $\hat{A}\subset B$. 
  
  By the isoperimetric inequality of the sphere, 
  $\mathrm{PP}_S(\hat{A})\geq
  \mathrm{PP}_S(A)$. Since map projections preserve containment,
  $\Sigma\subsetneq \varphi(B)$ implies that $A\subsetneq B$, 
  meaning that $\mathrm{area}(\hat A) = \mathrm{area}(A)\lneq 
  \mathrm{area}(B)$. By Corollary~\ref{cor:capscale}, we know that
  $\mathrm{PP_S}(\hat{A})< \mathrm{PP_S}(B)$, and combining this with
  the earlier inequality, we get
  \begin{align*}
    \mathrm{PP_S}({A})\leq \mathrm{PP_S}(\hat{A})< \mathrm{PP_S}(B)
  \end{align*}

  Since $\Sigma = \varphi(A)$ maximizes the Polsby-Popper score in the
  plane, but $A$ does not do so in the sphere, we have shown that
  $\varphi$ does not preserve the maximal elements in the score
  ordering, and therefore it cannot preserve the ordering itself.
\end{proof}

The reason why every map projection fails to preserve the ordering of
Polsby-Popper scores is because the score itself is constructed from
the \textit{planar} notion of isoperimetry, and there is no reason to
expect this formula to move nicely back and forth between the sphere
and the plane.  This proof crucially exploits a scale invariance
present in the plane but not the sphere.  If we consider any circle in
the plane, its Polsby-Popper score is equal to one,
but that is not true of every cap in the sphere.  \mute{This naturally
raises the question of whether being more careful, and defining
a compactness score which uses the isoperimetric quotient of the
surface the region is actually in will evade this problem.  We show
later that it does resolve the issue of scale-noninvariance, but it is
still induces an ordering which is not preserved by any map
projection.  We discuss this further in \Cref{sec:generalz}.}

\section{Empirical Results}\label{sec:exper}

While we have shown mathematically that the ordering of compactness scores can be permuted by any map 
projection, the actual districts we wish to examine are relatively small compared to the surface of the Earth, 
and we might ask whether this order reversal actually occurs in practice.
We can extract the boundaries of the districts from a \textit{shapefile} from the U.S. Census Bureau\footnote{We use the U.S. Census Bureau's shapefile for the Congressional districts for the 115th Congress.} and compute the convex hull, Reock, and Polsby-Popper scores on the Earth and with respect to a common map projection\footnote{The code to compute the various compactness scores is based on Lee Hachadoorian's \textit{compactr} project. \cite{hachadoorian2018reock}} and examine the ordering of the districts with respect to both.

The projection we use is the familiar \textit{Cartesian latitude-longitude} projection, which is the default projection used in the cartographic data provided by the US Census Bureau.  
The projection presents latitudes and longitudes as equally-spaced horizontal and vertical lines, which causes small amounts of distortion near the equator, but large amounts outside of the tropics, with the areas of regions near the poles being dramatically inflated under the projection.  Following \cite{hachadoorian2018reock}, we treat a local Albers equal-area projection as the ground-truth value for the compactness scores on the sphere, as computing the spherical values of these scores is not a simple task, even in modern GIS software.  To validate this assumption, we compare the numerical value of the Albers score to several other local projections, including the universal transverse Mercator and the local state plane projections and find they all agree up to at least five decimal places, well within the margin of error introduced by the discretization of the geographic data, as discussed in \cite{barnes2018gerrymandering}.

We observe overall that the orderings of districts under the Polsby-Popper score and convex hull score are relatively undisturbed compared to that of the Reock score.  This is because the \textit{values} of those scores do not change by too much under different projections.  Intuitively, this is because while both projections distort shapes, they do so in a way that does not affect either of these scores by too much. In the case of the Polsby-Popper score, the perimeter and area of the regions are changed in similar ways, and in the case of the convex hull score, the area of the convex hull of a region is distorted in the same way as the region itself. This leads to a similarity in the ratio of the areas of the districts across projections.

For regions in North America, such as our Congressional districts, the differences in the minimal bounding circle between the spherical and Cartesian representation cause massive differences in both the raw values of the scores as well as the ordering.  The scores can change in either direction by upwards of $.1$ (recall that the value of the score is between zero and one) and there is almost no correlation in the ordering of the districts by their Reock scores on the sphere and under the Cartesian projection. 
\begin{figure}[h]
	\centering
	\includegraphics[width=.5\textwidth]{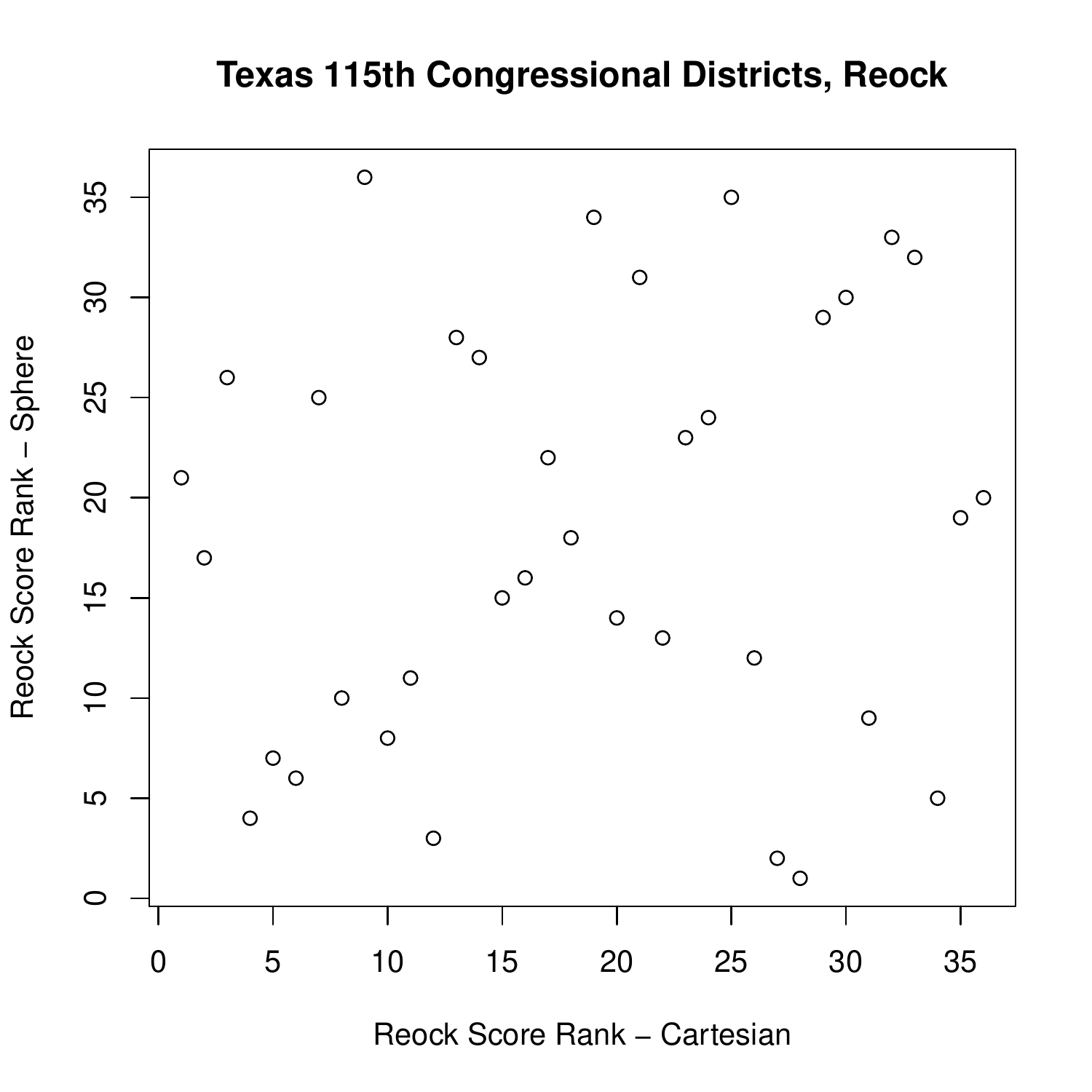}
	
	\caption{The permutation of the Reock score ordering for Texas' Congressional districts.}
	\label{fig:reock_exp_tx}
\end{figure}
In \Cref{fig:reock_exp_tx}, we plot the ordering of the 36 Congressional districts of Texas under the Reock score as computed on the sphere and in the plane after applying the Cartesian projection.  The rank on the sphere is along the horizontal axis and in the plane along the vertical axis.  Sweeping from left to right, one encounters the districts in order of least to most compact on the sphere and from bottom to top the least to the most compact in the plane.

A perfect preservation of the order would result in these points all falling on the diagonal.  However, what we see in practice is that many points do lie near the diagonal but several are very far away, indicating a strong disagreement between the Reock ordering on the sphere and in the plane.  This effect is not a result of some idiosyncrasy of the shapes of Texas' Congressional districts.  We find similar effects for all other states with at least a moderate number of districts.

This suggests that while the convex hull,  Polsby-Popper, and Reock measures all share a similar mathematical flaw, in applications to Congressional redistricting, the manifestation of the permutation of the ordering of districts by their Reock score is quite dramatic and could have real consequences for parties trying to use this score to assess the geometric compactness of a districting plan.  While we are unable to find any court cases where the Reock score was used crucially to determine whether or not a plan was an illegal gerrymander, the state laws of Michigan do use a similar measure in their definition of compactness, where one similarly constructs the smallest bounding circle of a district but then excludes from that circle any area falling outside of the state before taking the ratio.\footnote{Michigan Comp. Laws \S 3.63}  For districts whose bounding circle falls entirely within the state, this definition aligns exactly with the Reock score, and so is similarly susceptible to this permutation effect of map projections.

\section{Discussion}

We have identified a major \textit{mathematical} weakness in the commonly discussed compactness scores in that no map projection can preserve the ordering over regions induced by these scores.  This leads to several important considerations in the mathematical and popular examinations of the detection of gerrymandering.

From the mathematical perspective, rigorous definitions of compactness require more nuance than the simple score functions which assign a single real-number value to each district.  \textit{Multiscale} methods, such as those proposed in in \cite{deford2018tv}, assign a vector of numbers or a function to a region, rather than a single number.  The richer information contained in such constructions is less susceptible to perturbations of map projections.
Alternatively, we can look to capturing the geometric information of a district without having to work with respect to a particular spherical or planar representation.  So-called \textit{discrete compactness} methods, such as those proposed in \cite{duchin2018discrete}, extract a graph structure from the geography and are therefore unaffected by the choice of map projection, and our results suggest that this is an important advantage of these kinds of scores over traditional ones.  Finally, recent work has used lab experiments to discern what qualities of a region humans use to determine whether they believe a region is compact or not \cite{kingeyeball}.  Incorporating more qualitative techniques is important, especially in this setting where the social impacts of a particular districting plan may be hard to quantify.

We proved our non-preservation results for three particular compactness scores which appear frequently in the context of electoral redistricting.  There are countless other scores offered in legal codes and academic writing, such as definitions analogous to the Reock and convex hull scores which use different kinds of bounding regions, scores which measure the ratio of the area of the largest inscribed shape of some kind to the area of the district, and versions of these scores which replace the notion of \enquote*{area} with the population of that landmass.  Many of these and others suffer from similar flaws as the three scores we examined in this work.  It would be interesting to consider the most general version of this problem and enumerate a collection of properties such that any map projection permutes the score ordering of a pair of regions under a score with at least one of those properties. 

While compactness scores are not used critically in a \textit{legal} context, they appear frequently in the popular discourse about redistricting issues and frame the perception of the \enquote*{fairness} of a plan.  An Internet search for a term like `most gerrymandered districts' will invariably return results naming-and-shaming the districts with the most convoluted shapes rather than highlighting where more pleasant looking shapes resulted in unfair electoral outcomes. 

Similarly, a sizable amount of work towards remedying such abuses focuses primarily on the geometry rather than the politics of the problem. Popular press pieces (e.g.\ \cite{ingraham2014solve}) and academic research alike (e.g.\ \cite{voronoi, svec2007applying, levin_friedler_2019}) describe algorithmic approaches to redistricting which use geometric methods to generate districts with appealing shapes.  
However, these approaches ignore all of the social and political information which are critical to the analysis of whether a districting plan treats some group of people unfairly in some way. 
A purely geometric approach to drawing districts implicitly supposes that the mathematics used to evaluate the geometric features of districts are unbiased and unmanipulable and therefore can provide true insight into the fairness of electoral districts.  We proved here that the use of geographic compactness as a proxy for fairness is much less clear and rigid than some might expect.

This work opens several promising avenues for further investigation.  We prove strong results for the most common compactness scores, but the question remains what the most general mathematical result in this domain might be, such as giving a set of necessary and sufficient conditions for a compactness score to not induce a permuted order for some choice of map projection.

\ifarxiv
\subsection*{Acknowledgments}

 This work was partially completed while the authors were at the Voting Rights Data Institute in the summer of 2018, which was generously supported by the Amar G.\ Bose Grant. L.\ Najt was also supported by a grant from the NSF GEAR Network.

We would like to thank the participants of the Voting Rights Data Institute for many helpful discussions. Special thanks to Eduardo Chavez Heredia and Austin Eide for their help developing mathematical ideas in the early stages of this work. We would like to thank Lee Hachadoorian for inspiring the original research problem and Moon Duchin, Jeanne N. Clelland, and Anthony Pizzimenti for exceptionally helpful feedback on drafts of this work.  We thank Jeanne N. Clelland, Daryl DeFord, Yael Karshon, Marshall Mueller, Anja Randecker, and Caleb Stanford for offering their wisdom through helpful conversations throughout the process.

\else
\mute{
\begin{acknowledgment}{Acknowledgments}

 This work was partially completed while the authors were at the Voting Rights Data Institute in the summer of 2018, which was generously supported by the Amar G.\ Bose Grant. L.\ Najt was also supported by a grant from the NSF GEAR Network.

We would like to thank the participants of the Voting Rights Data Institute for many helpful discussions. Special thanks to Eduardo Chavez Heredia and Austin Eide for their help developing mathematical ideas in the early stages of this work. We would like to thank Lee Hachadoorian for inspiring the original research problem and Moon Duchin, Jeanne N. Clelland, and Anthony Pizzimenti for exceptionally helpful feedback on drafts of this work.  We thank Jeanne N. Clelland, Daryl DeFord, Yael Karshon, Marshall Mueller, Anja Randecker, and Caleb Stanford for offering their wisdom through helpful conversations throughout the process.
\end{acknowledgment}
}
\fi

\bibliography{bib}

%\pagebreak
%\appendix

%\input{curvature}

\end{document}